\documentclass[aps,prd,preprint,groupedaddress,nofootinbib]{revtex4-1}

\newcommand{\avg}[1]{\left< #1 \right>} 
\usepackage{graphicx}
\usepackage{epstopdf}
\usepackage[subrefformat=parens,labelformat=parens]{subfig}
\usepackage{rotating}
\usepackage[breaklinks=true,colorlinks=true,linkcolor=blue,urlcolor=blue,citecolor=blue]{hyperref}
\usepackage{amsmath,amsfonts,amsthm}
\newtheorem{theorem}{Theorem}

\usepackage[utf8]{inputenc}
\usepackage{tikz-feynman}
\tikzfeynmanset{compat=1.0.0}

\begin{document}




\begin{flushright}
	LU TP 19-08\\
	\today
\end{flushright}

\title{Anomaly-free Model Building with Algebraic Geometry}

\author{Johan Rathsman}
\email{johan.rathsman@thep.lu.se}
\author{Felix Tellander}
\email{felix@tellander.se}
\affiliation{Department of Astronomy and Theoretical Physics, Lund University,\\ SE-223 62, Lund, Sweden}



\date{\today}

\begin{abstract}
	We present a method to find anomaly-free gauged Froggatt-Nielsen type models using results from algebraic geometry. These methods should be of general interest for model building beyond the Standard Model (SM) when rational charges are required. We consider models with a gauged $U(1)$ flavor symmetry with one flavon and two Higgs doublets and three right-handed SM singlets to provide three model examples based on different physical assumptions. The models we study are: anomaly-free with no SM neutral heavy chiral fermions, anomaly-free with SM neutral heavy chiral fermions, and supersymmetric with SM neutral heavy chiral fermions where the anomalies cancel via the Green-Schwarz mechanism. With these different models we show how algebraic methods may be used in model building; both to reduce the charge constraints by calculation of Gröbner bases, and to find rational solutions to cubic equations using Mordell-Weil generators.
	
	Using these tools we find three phenomenologically viable models explaining the observed flavor structure.
\end{abstract}


\maketitle

\section{\label{introduction}Introduction}
The mass spectrum of fermions spans at least eleven orders of magnitude (from the top quark to the neutrinos) and if all these masses are to be generated by the Higgs mechanism, the couplings to the Higgs field must span an equal range. However, neutrino masses are often assumed to be generated from a seesaw mechanism \cite{MINKOWSKI1977,Yanagida:1979,Mohapatra1980,GellMann:1980}, or more generally, a Weinberg operator \cite{Weinberg1979}. If we adopt this explanation, the masses of the charged fermions masses still span six orders of magnitude. A famous explanation for this is the Froggatt-Nielsen (FN) mechanism \cite{Froggatt1979}. This provides an appealing explanation in terms of suppression factors $(\avg{S}/\Lambda_{FN})^n$, where $\Lambda_{FN}$ is the scale of integrated out physics, $\avg{S}$ the vacuum expectation value of the ``flavon" which breaks a new $U(1)$ gauge\footnote{The original FN-mechanism assumes that the symmetry is a global one, but here we will assume that it is local.} symmetry and $n$ depends on the charges of the fields under this new symmetry. A similar idea was also developed independently by Bijnens and Wetterich in \cite{Bijnens1987} but with heavy scalar fields instead of the heavy fermions used in \cite{Froggatt1979}. Throughout the paper we will use the term FN-mechanism independent of the origin of the suppression factors. In addition to the FN-mechanism, we will consider a two Higgs doublet model (2HDM), see \cite{Lee1973} and the review \cite{Branco2011}, where both Higgs fields are in general charged under the new $U(1)$ group. 

The main objective of this paper is to show how algebraic and Diophantine geometry provides powerful tools for finding rational flavon charges of the fermions and Higgs fields under this new $U(1)$ symmetry. 
The Standard Model gauge group is thus extended to $SU(3)_C\times SU(2)_L\times U(1)_Y\times U(1)'$ where $U(1)'$ denotes the new flavor dependent symmetry. 
Since the $U(1)'$ symmetry is local the
flavon charges have to cancel the triangle anomalies \cite{Bell1969,Adler1969a,Bardeen1969,Adler1969,Gross1969} of this gauge group, in addition to providing phenomenologically viable suppression factors for the fermion masses and reproducing the mixing matrices for the fermions. On top of this, we also want the flavon charges to be rational. Demanding rational charges is a significant challenge, since this then becomes related to Hilbert's 10th problem \cite{Hilbert1902} which is known to have no general solution \cite{Matiyasevich1973}. We proceed in this manner with a simple motivation. Since the two $U(1)$ gauge groups we know exist: $U(1)_Y$ and $U(1)_{EM}$, have rationally quantized charges, it is natural to assume that any new $U(1)$ should behave similarly. 


One of the best motivations for a 2HDM as an extension of the SM is its occurrence in supersymmetry (SUSY). However, vanishing anomalies and the Froggatt-Nielsen mechanism is contradictory in the SUSY setting \cite{IBANEZ1994,BINETRUY1995}. Therefore we instead
invoke the Green-Schwarz mechanism \cite{Green1984}, which is a string theoretic completion, to deal with anomaly cancellation in the case of SUSY. 

We have constructed a series of model examples to show how the algebraic methods may be used. These models have different mechanisms for anomaly cancellation and different phenomenological constraints imposed (we always demand recreation of the charged fermion masses and the CKM matrix):
\begin{itemize}
	\item (Sections \ref{model example} and \ref{sec: anomaly-free}) Here we study a 2HDM with three right-handed SM-neutral fermions where all anomalies vanish and neutrino masses are generated via the Weinberg operator.  
	\item  (Section \ref{sec: anomalous}) 2HDM particle content with three right-handed neutrinos where all anomaly coefficients vanish by the stated particle content, except the $U(1)'-U(1)'-U(1)'$ and graviton-graviton-$U(1)'$ anomalies, which are assumed to vanish by SM-neutral fermion content. The neutrino masses are generated by a type-I seesaw mechanism.
	\item (Section \ref{sec: supersymmetry}) Minimal supersymmetry with three right-handed neutrinos where the anomalies cancel via the Green-Schwarz mechanism and the neutrino masses are again generated by a type-I seesaw mechanism.
\end{itemize}

The paper is organized as follows. In Section \ref{sec: Froggatt-Nielsen} we review the Froggatt-Nielsen mechanism and derive the constraints from anomaly cancellation and flavor phenomenology. Next, in Section \ref{model example} we introduce the first model example and show how algebraic geometry naturally enters. Some aspects of algebraic geometry is discussed in Section \ref{sec: algebraic geometry} and the model example is then continued in Section \ref{sec: anomaly-free}. The two other model examples are given in Section \ref{sec: anomalous} and \ref{sec: supersymmetry}. Finally Section \ref{sec: conclusion} concludes the paper.

\section{Gauged Froggatt-Nielsen Mechanism in 2HDMs}\label{sec: Froggatt-Nielsen}
Let the SM gauge group be extended with a flavor dependent $U(1)$ symmetry denoted $U(1)'$. We assume that all fermions and both Higgs fields are charged under this new symmetry and call this charge flavon charge. Moreover, we assume that this symmetry is spontaneously broken when a complex scalar $S$, the flavon, with flavon charge $-1$, gets a vacuum expectation value (VEV).
Just above the energy scale $\avg{S}$, it is assumed that there exists many heavy vector-like fermion singlets, called FN-fermions, with mass $\sim\Lambda_{FN}$. 
At energies above $\avg{S}$ the observed fermions are effectively massless and the Yukawa couplings we observe in experiments are determined by physics at this scale where 
the heavy FN-fermions get their mass via a Higgs mechanism with a neutral scalar $\Phi'$. The different flavor properties of the fermions at the electroweak scale is encoded in different powers of the symmetry breaking parameter $\epsilon=\avg{S}/\Lambda_{FN}\approx 0.2$, which, following Froggatt and Nielsen \cite{Froggatt1979}, is chosen to fit the Wolfenstein parameterization \cite{Wolfenstein1983} of the CKM matrix. The powers of $\epsilon$ is then given by the number of flavon insertions needed for $U(1)'$ invariance. 

Assuming a 2HDM, the left-handed fermion fields we have are: $\{Q_L^i,(U_R^i)^c,(D_R^i)^c,L_L^i,(E_R^i)^c\}$ where $i=1,2,3$ is the flavor index and $(\cdot)^c$ denotes charge conjugation. In addition we also have the two Higgs fields $\{\Phi_1,\Phi_2\}$. We denote the flavon charges of these fields by $Q_i,u_i,d_i,L_i,e_i$ and $H_{1,2}$ respectively. Let us already here note that we will discuss physics at two different scales; the electroweak scale and the large $\Lambda_{FN}$ scale. In general these scales could be many orders of magnitude apart and one should therefore compare the physics at these scales using renormalization group evolution. This is, however, beyond the scope of the current paper. In any case, we do not expect large effects from this since the number of flavon insertions depends logarithmically on the masses.

The Yukawa Lagrangian in a general 2HDM is given by
\begin{equation}\label{eq: Lagrangian Yukawa 2HDM}
\begin{array}{ll}
-\mathcal{L}_Y=&\overline{Q}_L\widetilde{\Phi}_1Y_1^UU_R+\overline{Q}_L\Phi_1Y_1^DD_R+\overline{L}_L\Phi_1Y_1^LE_R\\
&+\overline{Q}_L\widetilde{\Phi}_2Y_2^UU_R+\overline{Q}_L\Phi_2Y_2^DD_R+\overline{L}_L\Phi_2Y_2^LE_R+\mathrm{H.c.}
\end{array}
\end{equation}
where $\widetilde{\Phi}=i\sigma_2\Phi^*$. As this Lagrangian stands, it is difficult to implement the Froggatt-Nielsen mechanism since we do not know a priori which of the Higgs fields provides the dominating mass contribution to each fermion. To circumvent this, and to remove flavor changing neutral currents (FCNCs) at tree-level, we impose a $\mathbb{Z}_2$-symmetry \cite{Glashow1977}. As is well known, there are four different ``Types" of $\mathbb{Z}_2$-symmetry as given in Table ~\ref{table: Z2} with the corresponding
$\mathbb{Z}_2$ charges.\footnote{In the models with right handed  neutrinos $N_R$, we assume that they have the same charges as $U_R$. In principle there are four more types for these models; the ones where $N_R$ has the opposite charge.}
\begin{table}
	\centering
	\caption{Different types of $\mathbb{Z}_2$ charge assignments for a 2HDM, the left-handed doublets $Q_L$ and $L_L$ are assigned ``+" in all cases.}
	\label{table: Z2}
	\begin{tabular}{lccccc}
		\hline
		$\mathbb{Z}_2$-symmetry &\quad $\Phi_1$ &\quad $\Phi_2$ &\quad $U_R$ &\quad $D_R$ &\quad $E_R$\\
		\hline
		Type-I (SM like) &\quad $+$ &\quad $-$ &\quad $-$ &\quad $-$ &\quad $-$\\
		Type-II (MSSM like) &\quad $+$ &\quad $-$ &\quad $-$ &\quad $+$ &\quad $+$\\
		Type-III/Y (flipped) &\quad $+$ &\quad $-$ &\quad $-$ &\quad $+$ &\quad $-$\\
		Type-IV/X (lepton specific) &\quad $+$ &\quad $-$ &\quad $-$ &\quad $-$ &\quad $+$\\
		\hline
	\end{tabular}
\end{table}

When the physics at the $\Lambda_{FN}$ scale is integrated out, the Yukawa couplings at the electroweak scale may be expressed as:
\begin{align}\label{eq: FN-mechanism Yukawa}
(Y_a^U)_{ij}\overline{Q}_L^i\widetilde{\Phi}_aU_R^j&\longrightarrow (g_a^U)_{ij}\left(\frac{\avg{S}}{\Lambda_{FN}}\right)^{|Q_i+u_j+H_a|}\overline{Q}_L^i\widetilde{\Phi}_aU_R^j\nonumber\\
(Y_a^D)_{ij}\overline{Q}_L^i\Phi_aD_R^j&\longrightarrow (g_a^D)_{ij}\left(\frac{\avg{S}}{\Lambda_{FN}}\right)^{|Q_i+d_j-H_a|}\overline{Q}_L^i\Phi_aD_R^j\\
(Y_a^L)_{ij}\overline{L}_L^i\Phi_aE_R^j&\longrightarrow (g_a^L)_{ij}\left(\frac{\avg{S}}{\Lambda_{FN}}\right)^{|L_i+e_j-H_a|}\overline{L}_L^i\Phi_aE_R^j\nonumber
\end{align}
where the $(g_a^F)_{ij}$ couplings are assumed to be $\sim\mathcal{O}(1)$ as in \cite{Froggatt1979}, with $F=U,D,L$ and $a=1,2$. The moduli in the exponents reflect the fact that we may choose either $S$ or $S^*$ to balance the flavon charges of the operators.

From the above structure of the Yukawa matrices, one must extract the masses and mixings. This is as usual done via bi-unitary transformations. Let us begin with the quark sector, assume that to each of the Yukawa matrices there is only one Higgs field providing the dominant mass contribution. Then the Yukawa matrices may be written as $Y^U_{ij}=g^U_{ij}\epsilon^{|Q_i+u_j+H_a|}$ and $Y^D_{ij}=g^D_{ij}\epsilon^{|Q_i+d_j-H_b|}$ where $a,b\in\{1,2\}$ are fixed. These matrices may now be written as
\begin{equation}
\begin{array}{rl}
Y^U&=(V_L^U)^\dagger D^UV_R^U\\
Y^D&=(V_L^D)^\dagger D^DV_R^D
\end{array}
\end{equation}
where $D^F,\ F=U,D$ are diagonal matrices. The philosophy of the FN-mechanism is that the magnitudes of the masses and mixings should solely depend on the $\epsilon$-parameters and thus one can take all the pre-factors $g$ to be of order one. For this to work one assumes that all the exponents in the Yukawa couplings are ordered\footnote{This is a crucial point and leads to that type-II 2HDM are preferred as shown in Section \ref{sec: sum rules}} such that 
\begin{equation}
\begin{array}{ll}
|Q_i+u_j+H_a|\ge|Q_{i+1}+u_j+H_a|,\ \ \ &|Q_i+u_i+H_a|\ge|Q_{i+1}+u_{i+1}+H_a|,\\
|Q_i+d_j-H_a|\ge|Q_{i+1}+d_j-H_a|,\ \ \ &|Q_i+d_i-H_a|\ge|Q_{i+1}+d_{i+1}-H_a|.
\end{array}
\end{equation}
Under these assumptions 
it is possible to diagonalize the Yukawa matrices analytically to leading order in $\epsilon$,
as shown in \cite{Froggatt1979}, giving:
\begin{equation}
\begin{array}{cc}
(V_L^U)_{ij}\sim\epsilon^{|Q_i-Q_j|},&\qquad (V_R^U)_{ij}\sim\epsilon^{|u_i-u_j|} \\
(V_L^D)_{ij}\sim\epsilon^{|Q_i-Q_j|},&\qquad (V_R^D)_{ij}\sim\epsilon^{|d_i-d_j|} 
\end{array}
\end{equation}
and the diagonal elements of the mass matrices are then given by
\begin{equation}
\begin{array}{c}
(D^U)_{ii}\sim\epsilon^{|Q_i+u_i+H_a|}\\
(D^D)_{ii}\sim\epsilon^{|Q_i+d_i-H_a|}
\end{array}
\end{equation}
{\it i.e.}~ the diagonal entries of $Y$. 
It then follows that the CKM-matrix is given by
\begin{equation}
(V_{CKM})_{ij}=(V_L^U)_{ik}(V_L^{D\ \dagger})_{kj}\sim\epsilon^{|Q_i-Q_j|} ,
\end{equation}
where we note that the mixing is to leading order determined by the flavon charges of the doublets.
For definiteness we will later assume without loss of generality that these charges are ordered, $Q_i\ge Q_{i+1}$.

A completely analogous calculation may be performed in the lepton sector once the mass matrix for the neutrinos is specified
yielding then also the mixing matrix for neutrinos, the so called PMNS-matrix.

\subsection{Neutrino masses}
In this paper we consider two ways of generating neutrino masses: directly via the Weinberg operator allowed by the FN-mechanism or via a type-I seesaw mechanism. Of course, when the right-handed fields in a type-I seesaw model are integrated out one obtains a Weinberg operator, but we still have to  distinguish between these two cases. To complete the symmetries between quarks and leptons, the neutrinos should have right-handed chiral partners. If the neutrino masses are generated by a Weinberg operator created solely by the FN-mechanism, then the right-handed fields have nothing to do with the mass generation so they only contribute to anomaly cancellation. Imposing a type-I seesaw is more restrictive; not only must the Yukawa couplings now be made gauge invariant, but the right-handed neutrinos must also effectively have Majorana masses. We will describe this in detail below.

Let us start with the case when the Weinberg operator is generated directly from the FN-mechanism. To the Lagrangian in Eq.~(\ref{eq: Lagrangian Yukawa 2HDM}) we must then add terms of the form
\begin{equation}\label{eq: Weinberg operator}
-\mathcal{L}_\nu^{(5)}=\frac{1}{2}\frac{(\kappa_{ab})_{ij}}{\Lambda_{FN}}\left(\widetilde{\Phi}_a^\dagger \overline{L_L^{c}}^j\right)\left(\widetilde{\Phi}_b^\dagger L_L^i\right)+\mathrm{H.c.}
\end{equation}
where $a,b\in\{1,2\}$. Imposing a $\mathbb{Z}_2$ symmetry restricts this term to $a=b$, but both Higgs fields may still contribute.

To generate this operator via the FN-mechanism the flavon charge in each of the two parenthesis must be an integer\footnote{It is sometimes stated in the literature that it is enough for them to be half-integers \cite{Appelquist2003}, but in a UV completion with vector-like fermions the chiralities will not add up unless the flavon charge is an integer.} so that the middle transition in the generating diagram, labeled $\chi$ in Fig. \ref{fig: Weinberg operator FN}, is made by an {\emph{uncharged}} Majorana fermion. The couplings at the electroweak scale may now be expressed as
\begin{equation}
(\kappa_{aa})_{ij}\longrightarrow (\kappa_{aa}^{\nu})_{ij}\left(\frac{\avg{S}}{\Lambda_{FN}}\right)^{|L_i+H_a|+|L_j+H_a|}
\end{equation}
where $a=1,2$ (i.e. we assume a $\mathbb{Z}_2$-symmetry), both the moduli have to be integers and $(\kappa_{aa}^\nu)_{ij}\sim\mathcal{O}(1)$ in the FN-spirit.

We assume here that the Majorana fermion $\chi$ is one of the FN fermions so that it also has a mass $\sim\Lambda_{FN}$. For this operator to not only account for the hierarchies but also the overall smallness of the neutrino masses, $\Lambda_{FN}$ has to be of the order $10^{14}$ GeV. Otherwise the $L_i$ and $H_a$ charges have to be increased accordingly.

The flavon VEV, $\avg{S}$, must be of the same order as $\Lambda_{FN}$ and thus the mass of the $Z'$ boson associated with $U(1)'$ must also be very large (if it is not extremely weakly coupled). This is readily seen from the relation
\begin{equation}
m_{Z'}\approx g_{Z'}\avg{S}
\end{equation}
where $g_{Z'}$ is the gauge coupling of the $Z'$. For this type of model, $Z'$ phenomenology is therefore not interesting, either the $Z'$ boson is so massive that its effects are unobservable, or if its mass scale is reachable by todays experiments, it has to be so weakly coupled that its effects would still be unobservable.   
\begin{figure}
	\centering
	\includegraphics[width=0.8\textwidth]{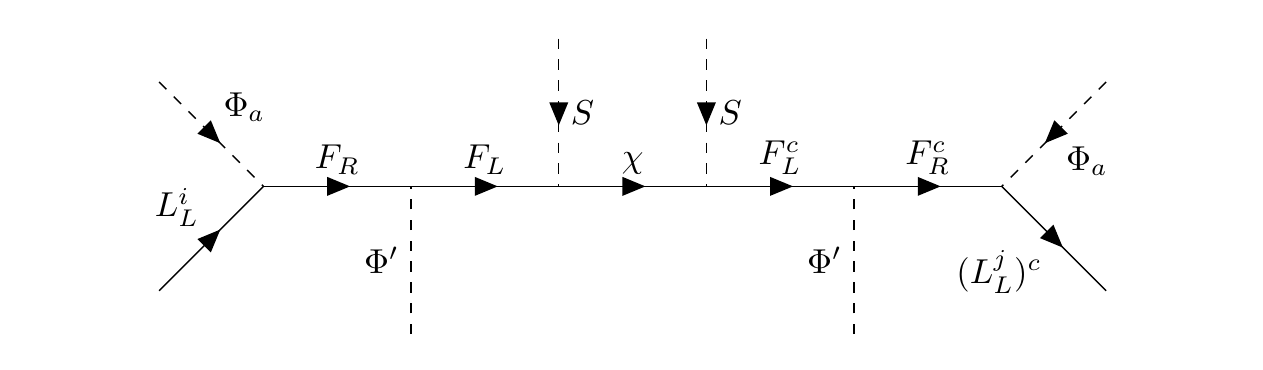}
	\caption{A diagram generating the $ij$-element of the Weinberg operator, $\chi$ is a Majorana fermion with mass $\sim\Lambda_{FN}$. Here the four-component spinor Feynman rules from ref.~\cite{Dreiner2010} are used.}
	\label{fig: Weinberg operator FN}
\end{figure}

If we instead want to use a type-I seesaw mechanism to generate the neutrino masses, we have to introduce the three $SU(2)_L$ singlet fields $N_R^i$, $i=1,2,3$, where we denote the flavon charge of the left-handed field $(N_R^i)^c$ by $\nu_i$. To the Lagrangian in Eq.~(\ref{eq: Lagrangian Yukawa 2HDM}) we must then add the terms
\begin{equation}\label{eq: L_N}
-\mathcal{L}_N=\overline{L}_L\widetilde{\Phi}_1Y_1^N N_R+\overline{L}_L\widetilde{\Phi}_2Y_2^N N_R+\frac{1}{2}M_R\overline{N_R^{c}}N_R+\mathrm{H.c.}
\end{equation}  
and if a $\mathbb{Z}_2$-symmetry is imposed it will only be one of the Yukawa terms that generates Dirac masses. 

The FN-mechanism for the Yukawa terms works the same way as for the terms in Eq.~(\ref{eq: FN-mechanism Yukawa}), so we have
\begin{equation}
(Y_a^N)_{ij}\longrightarrow (g_a^N)_{ij}\left(\frac{\avg{S}}{\Lambda_{FN}}\right)^{|L_i+\nu_j+H_a|}
\end{equation}
with $(g_a^N)_{ij}\sim\mathcal{O}(1)$. The Majorana masses for the right-handed fields may also be generated by the FN-mechanism:
\begin{equation}
\frac{1}{2}(M_R)_{ij}\overline{N_R^{c}}^iN_R^j\longrightarrow\frac{1}{2}\Lambda_{FN}(g^R)_{ij}\left(\frac{\avg{S}}{\Lambda_{FN}}\right)^{|\nu_i|+|\nu_j|}\overline{N_R^{c}}^iN_R^j
\end{equation}
where $(g^R)_{ij}\sim\mathcal{O}(1)$ and both $|\nu_i|$ and $|\nu_j|$ have to be integers so that a diagram similar to Fig. \ref{fig: Weinberg operator FN} may be drawn with an uncharged Majorana fermion doing the transition in the middle of the diagram. 

With the Dirac masses given by $m_D=(v_a/\sqrt{2})Y_a^N$ and the Majorana masses $M_R$ as just discussed, the light physical neutrino masses are given by (assuming $m_D \ll M_R$)
\begin{equation}
m_\nu=-m_D(M_R)^{-1}m_D^T.
\end{equation}
Since $M_R\sim\Lambda_{FN}$, we have $m_\nu\sim v^2/\Lambda_{FN}$, so that, just as in the case with the Weinberg operator, $\Lambda_{FN}$ must be of order $10^{14}$ GeV to account for the smallness of the neutrino masses (unless $|\nu_i|+|\nu_j|\sim 20$).

\subsection{Anomaly cancellation}
An important aspect of a gauged Froggatt-Nielsen mechanism is that the flavon charges not only have to fit with the phenomenological constraints, but also have to satisfy anomaly constraints. For the gauge group $SU(3)_C\times SU(2)_L\times U(1)_Y\times U(1)'$ together with gravity, there are six triangle diagrams whose contributions do not cancel trivially. In the following, let $\mathcal{A}_{XYZ}=\frac{1}{2}\mathrm{tr}[T_X\{T_Y,T_Z\}]$ where $T_X$ are the generators of the gauge group $X$ in the fundamental representation. For hypercharge we adopt the normalization that $Y=2(Q-T_3)$. The six anomaly constraints involving the $U(1)'$-charges are then given by
\begin{equation}\label{eq: anomalies U1}
\begin{array}{ll}
\mathcal{A}_{11'1'}&=2 {\displaystyle \sum_{j=1}^3}\left(Q_j^{2}-2u_j^{2}+d_j^{2}-L_j^{2}+e_j^{2}\right)=0\\
\mathcal{A}_{111'}&=\dfrac{2}{3} {\displaystyle\sum_{j=1}^{3}}\left(Q_j+8u_j+2d_j+3L_j+6e_j\right)=0\\
\mathcal{A}_{331'}&=\dfrac{1}{2}{\displaystyle\sum_{j=1}^3}\left(2Q_j+u_j+d_j\right)=0\\
\mathcal{A}_{221'}&=\dfrac{1}{2}{\displaystyle\sum_{j=1}^3}\left(3Q_j+L_j\right)=0\\
\mathcal{A}_{1'1'1'}&= {\displaystyle\sum_{j=1}^3}\left(6Q_j^{3}+3u_j^{3}+3d_j^{3}+2L_j^{3}+e_j^{3}+\nu_j^3\right)=0\\
\mathcal{A}_{gg1'}&={\displaystyle\sum_{j=1}^3}\left(6Q_j+3u_j+3d_j+2L_j+e_j+\nu_j\right)=0
\end{array}
\end{equation}
where $\mathcal{A}_{gg1'}$ is from the triangle diagram with two gravitons and one $U(1)'$ boson. Note that the gravitational anomaly may be written as $\mathcal{A}_{gg1'}=6\mathcal{A}_{331'}+\sum_{j=1}^{3}(2L_j+e_j+\nu_j)$ so when implemented later we only need to care about the leptonic part.
\subsection{Sum rules for FN-constraints}\label{sec: sum rules}
In this section we will derive a set of sum rules that show how the Froggatt-Nielsen constraints for each imposed $\mathbb{Z}_2$-symmetry are related to the anomaly constraints. If these rules are not satisfied, there will not exist an anomaly-free charge assignment satisfying the imposed FN-constraints. This generalizes some of the results in \cite{IBANEZ1994,BINETRUY1995} where SUSY was considered to more general 2HDMs. In addition, we show that these rules imply that the type-II symmetry is favored by the FN-mechanism since the other symmetries will lead to skewed Yukawa matrices with large off-diagonal elements. This is problematic since in the FN-mechanism, it is assumed that the diagonal elements in the Yukawa matrices directly gives the masses and large off-diagonal elements will spoil the diagonalization such that this is no longer the case.

Let us start with a type-II symmetry and denote by $\{n_u,n_c,n_t,n_d,n_s,n_b,n_e,n_\mu,n_\tau\}$ the signed number of $\epsilon$-factors suppressing the masses, e.g. the up quark mass is suppressed by $\epsilon^{|n_u|}$. With a type-II symmetry, we know that $n_u=Q_1+u_1+H_2$ and so on. Now, using the two sets of fermions that couple to the same Higgs field, which in this case are the down-type quarks and $e, \mu, \tau$-leptons, we obtain the following sum rule:
\begin{equation}\label{eq: sum rule type-II}
n_d+n_s+n_b-n_e-n_\mu-n_\tau=\sum_{j=1}^3(Q_j+d_j-L_j-e_j)=\frac{8}{3}\mathcal{A}_{331'}-\frac{1}{4}\mathcal{A}_{111'}-\mathcal{A}_{221'}=0.
\end{equation}

Similarly in the type-Y (flipped) case, using the fermions that couple to the same Higgs field, i.e. the up-type quarks and $e, \mu, \tau$-leptons, we obtain the following rule:
\begin{equation}\label{eq: sum rule type-III}
n_u+n_c+n_t+n_e+n_\mu+n_\tau=\sum_{j=1}^3(Q_j+u_j+L_j+e_j)=-\frac{2}{3}\mathcal{A}_{331'}+\frac{1}{4}\mathcal{A}_{111'}+\mathcal{A}_{221'}=0.
\end{equation} 

With type-X (lepton specific) symmetry it is the two sets of quarks that couple to the same Higgs field, this yields the rule:
\begin{equation}\label{eq: sum rule type-IV}
n_u+n_c+n_t+n_d+n_s+n_b=\sum_{j=1}^3(2Q_j+u_j+d_j)=2\mathcal{A}_{331'}=0.
\end{equation}
Finally, for type-I (SM-like) symmetry, all three rules: Eqs.~(\ref{eq: sum rule type-II}-\ref{eq: sum rule type-IV}), have to be satisfied. Two of the constraints above imply the third, so in practice, the SM-like 2HDM only gets two constraints from the sum rules and not three.

The sum rules that do not have to be satisfied for a given $\mathbb{Z}_2$ symmetry still affect the flavon charges since they specify the charges of the Higgs fields. For example, given a type-II model, Eq. (\ref{eq: sum rule type-IV}) gives
\begin{equation}
n_u+n_c+n_t+n_e+n_\mu+n_\tau=2\mathcal{A}_{331'}+3(H_2-H_1)=3(H_2-H_1)\ \in\mathbb{Z}
\end{equation} 
where $\mathbb{Z}$ denotes the integers.  
The same constraint is of course obtained from Eq. (\ref{eq: sum rule type-III}). This means that $H_2-H_1$ is specified by the suppression factors. 
In addition we see that, $H_2-H_1 \in\mathbb{Z}/3$ so that this difference may be an integer depending on the suppression factors.

As promised above, we will now argue that these sum rules imply that type-II symmetry is favored by the FN-mechanism. The reason is simple, it is all due to the minus signs on the left hand side in Eq.~(\ref{eq: sum rule type-II}) between the down-quarks and the $e, \mu, \tau$-leptons. These minus signs allow for $n_d,n_s,n_b,n_e,n_\mu,n_\tau>0$ and still satisfying the sum rule, while for the other rules, at least one of the $n$'s has to be smaller than zero. When this happens, the Yukawa matrices are prone to be skewed with 
off-diagonal elements breaking the $\epsilon$-ordering assumed when diagonalizing the mass-matrices 
(more specifically $|Q_i+u_j+H_a|\ge|Q_{i+1}+u_j+H_a|$ is broken if $Q_i+u_j+H_a < 0$ for ordered $Q_i$). 
As a consequence,
the fermion masses will no longer correspond to the diagonal elements. In other words, the idea behind the FN-mechanism does not apply and we therefore consider these situations disfavored. In what follows, we will therefore always assume a type-II symmetry when we impose the FN-constraints. 
For simplicity we will also assume that $\tan\beta=1$. Other values of $\tan\beta$ can easily be incorporated by reducing  $n_d,n_s,n_b,n_e,n_\mu,n_\tau$ accordingly.

Adding neutrino Yukawa couplings provide additional sum rules, but the conclusion above is unaffected by this.

\section{Anomaly-free Model example}\label{model example}
We now consider a type-II 2HDM with $\tan\beta=1$ and neutrino masses generated via the Weinberg operator.
As already mentioned, we assume three right-handed SM-neutral fermions in addition to the normal 2HDM particle content and they are assumed to have flavon charges $\nu_i$.
The FN-constraints we impose are:
\begin{equation}\label{eq: Yukawa constraints}
Y_2^U\sim\begin{pmatrix}
\epsilon^7 & \epsilon^5 & \epsilon^3\\
\epsilon^6 & \epsilon^4 & \epsilon^2\\
\epsilon^4 & \epsilon^2 & \epsilon^0
\end{pmatrix},\ \ \ Y_1^D\sim\begin{pmatrix}
\epsilon^7 & \epsilon^6 & \epsilon^6\\
\epsilon^6 & \epsilon^5 & \epsilon^5\\
\epsilon^4 & \epsilon^3 & \epsilon^3
\end{pmatrix},\ \ \ Y_1^L\sim\begin{pmatrix}
\epsilon^8 & * & *\\
* & \epsilon^4 & \epsilon^3\\
* & \epsilon^4 & \epsilon^3
\end{pmatrix}
\end{equation}
for the Yukawa matrices and
\begin{equation}\label{eq: Weinberg constraints}
\kappa_{11}\sim\begin{pmatrix}
* & * & *\\
* & \epsilon^0 & \epsilon^0\\
* & \epsilon^0 & \epsilon^0
\end{pmatrix}
\end{equation}
for the Weinberg operator generated with the $\Phi_1$-field\footnote{In principle one could also include a Weinberg operator with  $\Phi_2$, but for this case there exist no rational solutions for the flavon charges. Similarly it does not exist any rational solutions if we try to include Dirac or Majorana masses for the right-handed SM-neutral fermions.}
where * denotes an element we do not determine a priori. 
The $\epsilon$-suppression of the masses may be read off from the diagonal and the off-diagonal elements of the quark Yukawa matrices guarantees a CKM-matrix on the form
\begin{equation}
V_{CKM}\sim\begin{pmatrix}
1 & \epsilon & \epsilon^3\\
\epsilon & 1 & \epsilon^2\\
\epsilon^3 & \epsilon^2 & 1
\end{pmatrix}
\end{equation}
while the off-diagonal elements in the lepton Yukawa matrix and the Weinberg operator guarantees large $\nu_\mu-\nu_\tau$ mixing (assuming normal neutrino mass hierarchy). When writing down the constraints for $Y_1^D$ and $Y_1^L$ we have made sure that the sum rule for type-II models, Eq.~(\ref{eq: sum rule type-II}), is satisfied.

Even though it is not necessary for this model, since we have either an extremely massive or weakly coupled $Z'$, we may remove mixing between $U(1)_Y$ and $U(1)'$ in the massless limit by adding
\begin{equation}\label{eq: kinetic mixing}
\sum_{j=1}^3 (2Q_j-4u_j+2d_j-2L_j+2e_j)=0
\end{equation}
to the list of constraints. This is just the trace of the hyper charge and flavon charge generators.

All the phenomenological constraints: Eqs.~(\ref{eq: Yukawa constraints}), (\ref{eq: Weinberg constraints}) and (\ref{eq: kinetic mixing}), and the anomaly conditions, Eq.~(\ref{eq: anomalies U1}), are summarized in the following system of polynomial equations:
\begin{equation}\label{eq: polynomial system}
\begin{cases}
\sum_{j=1}^3\left(Q_j^{2}-2u_j^{2}+d_j^{2}-L_j^{2}+e_j^{2}\right)=0\\
\sum_{j=1}^{3}\left(Q_j+8u_j+2d_j+3L_j+6e_j\right)=0\\
\sum_{j=1}^3\left(2Q_j+u_j+d_j\right)=0\\
\sum_{j=1}^3\left(3Q_j+L_j\right)=0\\
\sum_{j=1}^3\left(6Q_j^{3}+3u_j^{3}+3d_j^{3}+2L_j^{3}+e_j^{3}+\nu_j^3\right)=0\\
\sum_{j=1}^3\left(2L_j+e_j+\nu_j\right)=0\\
\sum_{j=1}^3 (2Q_j-4u_j+2d_j-2L_j+2e_j)=0\\
Q_3+u_3+H_2=0,\ 
Q_2+u_2+H_2=4,\ 
Q_1+u_1+H_2=7\\ 
Q_3+d_3-H_1=3,\ 
Q_2+d_2-H_1=5,\
Q_1+d_1-H_1=7\\ 
L_3+e_3-H_1=3,\ 
L_2+e_2-H_1=4,\ 
L_1+e_1-H_1=8\\
Q_1-Q_2=1,\
Q_2-Q_3=2\\  
L_2-L_3=0,\
L_2+H_1=0
\end{cases}
\end{equation}
To find flavon charges that satisfies this system we will proceed by using Gröbner bases and methods from Diophantine geometry.
\section{Algebraic geometry}\label{sec: algebraic geometry}
In this section we discuss some general aspects of algebraic geometry and give some results useful for finding rational charges. A more detailed, but still short description may be found in \cite{Tellander2018}. 

The first tool we want to mention is the key notion of computational algebraic geometry, that of Gröbner bases. The Gröbner basis of a system of equations may be thought of as the most reduced version of the system, similar to putting a linear system of equations on echelon form. As for a linear system on echelon form, a Gröbner basis 
with a given
lexicographic ordering has the property that once the last equation is solved, all other equations may be solved by back-substitution. Another useful property is that a system has no solution if and only if 1 is in the Gröbner basis. 

As for calculating the Gröbner basis in practice, the exact method is in general of no interest for such applied problems we study here, so one can without worry use it as a black-box command in e.g. Sage \cite{sagemath} or Macaulay 2 \cite{M2} (it is also implemented in some general purpose programs such as Maple \cite{maple} and Mathematica \cite{mathematica}). 

The set of solutions to a system of polynomials is called a variety. If the variety is zero-dimensional, i.e. consists of points, then the cubic and quadratic equations from the anomaly conditions make it unlikely that these points would be rational. To find rational points, it is therefore in general best to choose the number of linear constraints to implement such 
that the variety becomes one-dimensional, i.e. a curve
\footnote{At the same time, if it turns out that this curve is linear in one of the charges and has no dependence on the others, then one can add one more linear constraint giving a point solution. We will see two such special cases below.}. 
There is a rich literature on finding rational points on algebraic curves, from which we will discuss a few of the results below.

Let $C$ be a curve with rational coefficients defined by a polynomial equation $P(x,y)=0$, we call this an affine curve and we denote the set of rational points $C(\mathbb{Q})$ where $\mathbb{Q}$ denotes the rational numbers. The corresponding \emph{projective} curve is defined by $Z^{\mathrm{deg}P}P(X/Z,Y/Z)$, such that all terms in the polynomial has the same total degree in $X,Y,Z$, and we assume without loss of generality (see \cite{Fulton1969} Section 7.5 Theorem 3) that it is smooth. Smooth curves satisfy the following trichotomy classified by the genus $g$:
\begin{itemize}
	\item $g=0$:\\
	Here we have two choices: either $C(\mathbb{Q})=\emptyset$ or $C(\mathbb{Q})$ is non-empty which means that $C$ is isomorphic over $\mathbb{Q}$ to the projective line $\mathbb{P}^1$. Any such isomorphism defines a parameterization of $C(\mathbb{Q})$ in terms of rational functions in one variable, which is easily computable. For example, all rational points on the unit circle $x^2+y^2=1$ are given by
	\begin{equation}
	(x(t),y(t))=\left(\frac{1-t^2}{1+t^2},\frac{2t}{1+t^2}\right)
	\end{equation} 
	for $t\in\mathbb{P}^1(\mathbb{Q})=\mathbb{Q}\cup\{\infty\}$.
	\item $g=1$:\\
	For this case we have the following theorem:
	\begin{theorem}\label{theorem: Mordell-Weil}
		\textbf{Mordell-Weil:} For any Abelian variety the set of $K$-rational points form a finitely generated group.
	\end{theorem}
	\begin{proof}
		For the original proof for elliptic curves by Mordell, see \cite{Mordell1922}, and for the generalization to Abelian varieties by Weil, see \cite{Weil1929}.
	\end{proof}
	For $K=\mathbb{Q}$ this means that the only genus one curves with rational points are the elliptic curves.
	\item $g\ge 2$:\\
	For these higher genus curves, Mordell \cite{Mordell1922} conjectured and Falting \cite{Faltings1983} later proved that the set of $K$-rational points is finite.
\end{itemize}
For genus zero curves, the rational parameterization (if it exists) is easily obtained using the programs we have already mentioned. There are also well-developed methods to find integer solutions, see refs. \cite{Poulakis2000,Poulakis2002}.

In the case of genus one curves,  
we know by the Mordell-Weil theorem that the set of rational points on an elliptic curve form a finitely generated Abelian group, denoted $E(\mathbb{Q})$. 
The structure theorem then tells us that
\begin{equation}
E(\mathbb{Q})=E(\mathbb{Q})_\mathrm{tors}\oplus\mathbb{Z}P_1\oplus\ldots\oplus\mathbb{Z}P_r
\end{equation}
where $E(\mathbb{Q})_\mathrm{tors}$ is the finite subgroup of $E(\mathbb{Q})$ consisting of all elements of finite order and $r$ is the rank of $E(\mathbb{Q})$. There is no known algorithm to determine the rank $r$ or to find the Mordell-Weil generators $P_1,\ldots,P_r$ in general.

For curves of genus at least two it is harder to find rational points. However, point search might turn out to be more successful than for elliptic curves since the rational points are expected to have smaller height for curves with higher genus 
\cite{Stoll2011}. Here the height of a point $P(X:Y:Z)$, where $X,Y$ and $Z$ are integers with no common factors, is given by $\max\{|X|,|Y|,|Z|\}$.

In our type of models we have one cubic $(\mathcal{A}_{1'1'1'})$ and one quadratic $(\mathcal{A}_{11'1'})$ equation while the rest are linear. The typical degree of the variety is therefore six. 
However, given additional fermions only charged under $U(1)'$ and not under the SM groups, $\mathcal{A}_{11'1'}$ will still depend only on the SM fields 
whereas $\mathcal{A}_{1'1'1'}$ depends on the additional fields. 
In such a case the Gröbner basis may decouple into two parts that can be solved independently if there are enough linear constraints. Thus, $\mathcal{A}_{11'1'}$ may be solved independently from $\mathcal{A}_{1'1'1'}$ and we expect a solution of degree at most three for the latter one. Similarly if the cubic constraint is not applied we expect a solution of degree two at most. 
As a curve of degree three (two) has at most genus one (zero), the above methods are typically enough to now go back and solve our system in Eq.~(\ref{eq: polynomial system}).

\section{Anomaly-free Model Example, continued}\label{sec: anomaly-free}
Using Sage we find that 
the Gröbner basis for the system in Eq.~(\ref{eq: polynomial system}) is given by
\begin{equation}\label{eq: example groebner basis}
\begin{array}{rrr}
Q_1 - 8/27=0,&\qquad Q_2 + 19/27=0,&\qquad Q_3 + 73/27=0,\\
u_1 - 26/27=0,&\qquad u_2 + 28/27=0,&\qquad u_3 + 82/27=0,\\
d_1 - 34/9=0,&\qquad d_2 - 25/9=0,&\qquad d_3 - 25/9=0,\\
L_1 - 94/27=0,&\qquad L_2 - 79/27=0,&\qquad L_3 - 79/27=0,\\
e_1 - 43/27=0,&\qquad e_2 + 50/27=0,&\qquad e_3 + 77/27=0,\\
H_1 + 79/27=0,&\qquad H_2 - 155/27=0,&\qquad 
\end{array}
\end{equation}
and
\begin{align}\label{eq: example equation}
&\nu_1 + \nu_2 + \nu_3 + 140/9=0,\nonumber\\
&\nu_2^2\cdot \nu_3 + 140/9\cdot \nu_2^2 + \nu_2\cdot \nu_3^2 + 280/9\cdot \nu_2\cdot \nu_3 + 19600/81\cdot \nu_2 +\\
&+ 140/9\cdot \nu_3^2 + 19600/81\cdot \nu_3 + 95036/81=0.\nonumber
\end{align}
In this case the Gröbner basis has decoupled into two parts, as discussed above, 
with the flavon charges of the three right-handed singlets determined by the cubic and gravitational anomalies whereas the flavon charges of all the SM fields are determined by the other anomalies and the FN constraints. The only connection between the two is that the flavon charges of the SM fields feeds into the numerical constants in Eq.~(\ref{eq: example equation}).
All flavon charges can thus be directly read off from the Gröbner basis except those of the three right-handed singlets. 
Note that $H_2-H_1=26/3\notin\mathbb{Z}$ so the initially imposed type-II $\mathbb{Z}_2$-symmetry is in this case a residual effect from $U(1)'$ invariance. 
It should also be noted that if the non-mixing constraint  Eq.~(\ref{eq: kinetic mixing}) is not applied, then this would lead to one of the charges in Eq.~(\ref{eq: example groebner basis}) to be left free, the solution being given by a genus zero curve, and the other charges would be linear functions of it. 

To find a complete set of rational charges, we have to solve the equations for the $\nu_i$'s. To do this, we begin to study the cubic equation in Eq.~(\ref{eq: example equation}). 
This is a smooth curve of degree three so its genus is one 
and thereby it is an elliptic curve. 
We may thus hope to calculate the Mordell-Weil generators. 
The starting point is to write the curve on Weierstrass form. To do this we 
first write the curve on its projective form by introducing the homogenizing variable $h$. Next we map the projective version according to:
\begin{equation}\label{eq: map to Weierstrass}
(\nu_2:\nu_3:h)\mapsto(X:Y:Z)=\left(h : -\nu_3 - h : \frac{81}{95036}\nu_2 + \frac{81}{95036}\nu_3\right)
\end{equation} 
Given the form of Eq.~(\ref{eq: example equation}) this mapping can always be found. 
Finally, changing variables to $x=X/Z$ and $y=Y/Z$ gives the curve on Weierstrass form:
\begin{equation}
E:\ y^2 + 2xy +\frac{95036}{81}y = x^3 +\frac{19519}{81}x^2 + \frac{12449716}{729}x.
\end{equation}  
We do not find rational solutions to this curve directly, but after
 doing two-descent in Sage, 
it is found that this curve has rank one which means that the set of rational points is given by
\begin{equation}
E(\mathbb{Q})=E(\mathbb{Q})_{\mathrm{tors}}\oplus\mathbb{Z}P_1.
\end{equation} 
Explicitly we find
\begin{align}
E(\mathbb{Q})_\mathrm{tors}&=\{(0 : -95036/81 : 1), (0 : 0 : 1), (0 : 1 : 0)\}\nonumber\\
P_1&=(2041940/81 : 323674124/81 : 1)
\end{align}
and by mapping the point $P_1$ back to the original curve in Eq.~(\ref{eq: example equation}), by inverting Eq. (\ref{eq: map to Weierstrass}) such that $\nu_3=-\dfrac{X+Y}{X}$ and $\nu_2= \dfrac{95036}{81}\dfrac{Z}{X}-\nu_3$, 
we get
\begin{equation}\label{eq: example rational solution}
(\nu_2,\nu_3)=\left(\frac{30795}{193},-\frac{18344}{115}\right).
\end{equation} 
The last charge is now simply determined by
\begin{equation}\label{eq: example R1 charge}
\nu_1=-\frac{140}{9}-\nu_2-\nu_3=-\frac{3116597}{199755}
\end{equation}
Out of all the points generated by $P_1$ this one has the smallest height we have found.

\begin{table}
	\centering
	\caption{An example of rational charges satisfying Eq.~(\ref{eq: polynomial system}).}
	\label{table: examples charges}
	\begin{tabular}{lrrrrrc}
		\hline
		Generation $i$ &\quad $Q_i$ &\quad $u_i$ &\quad $d_i$ &\quad $L_i$ &\quad $e_i$ &\quad $\nu_i$\\
		\hline
		1 &\quad $\frac{8}{27}$ &\quad $\frac{26}{27}$ &\quad $\frac{34}{9}$ &\quad $\frac{94}{27}$ &\quad $\frac{43}{27}$ &\quad $-\frac{3116597}{199755}$\\
		2 &\quad $-\frac{19}{27}$ &\quad $-\frac{28}{27}$ &\quad $\frac{25}{9}$ &\quad $\frac{79}{27}$ &\quad $-\frac{50}{27}$ &\quad $\frac{30795}{193}$\\
		3 &\quad $-\frac{73}{27}$ &\quad $-\frac{82}{27}$ &\quad $\frac{25}{9}$ &\quad $\frac{79}{27}$ &\quad $-\frac{77}{27}$ &\quad $-\frac{18344}{115}$ \\  
		\hline
		Higgs charges: &\quad $H_1=-\frac{79}{27}$ &\quad $H_2=\frac{155}{27}$ & & & &\\
		\hline
	\end{tabular}
\end{table}

To summarize the results in this section, the complete set of flavon charges for this model is shown in Table \ref{table: examples charges} and the Yukawa matrices and mixings they produce are:
\begin{align}\label{eq: example couplings}
&Y_2^U\sim\begin{pmatrix}
\epsilon^7 & \epsilon^5 & \epsilon^3\\
\epsilon^6 & \epsilon^4 & \epsilon^2\\
\epsilon^4 & \epsilon^2 & \epsilon^0
\end{pmatrix},\ \ \ &&Y_1^D\sim\begin{pmatrix}
\epsilon^7 & \epsilon^6 & \epsilon^6 \\
\epsilon^6 & \epsilon^5 & \epsilon^5\\
\epsilon^4 & \epsilon^3 & \epsilon^3
\end{pmatrix}\ \ \ &&V_{CKM}\sim\begin{pmatrix}
1 & \epsilon & \epsilon^3\\
\epsilon & 1 & \epsilon^2\\
\epsilon^3 & \epsilon^2 & 1
\end{pmatrix}\nonumber\\
&Y_1^L\sim\begin{pmatrix}
\epsilon^8 & 0 & 0\\
0 & \epsilon^4 & \epsilon^3\\
0 & \epsilon^4 & \epsilon^3
\end{pmatrix}\ \ \ && \kappa_{11}\sim\begin{pmatrix}
0 & 0 & 0\\
0 & 1 & 1\\
0 & 1 & 1
\end{pmatrix}\ \ \ && U_{PMNS}\sim\begin{pmatrix}
1 & 0 & 0\\
0 & 1 & 1\\
0 & 1 & 1
\end{pmatrix}.
\end{align}
As can be seen from the resulting matrices, all elements which were left undetermined turn out to be zero due to $U(1)'$ invariance. This also means that in the resulting PMNS-matrix, there is no neutrino oscillation with the first generation. Apart from that, these 
matrices reproduce all observed flavor phenomenology.

\section{Models Requiring New Chiral Fermions}\label{sec: anomalous}
Assuming that the $\mathcal{A}_{1'1'1'}$ and $\mathcal{A}_{gg1'}$ anomalies vanish is equivalent to assuming that there either are no unknown SM-neutral fermions, or that the unknown SM-neutral fermions cancel the anomalies independently. However, since we in reality know nothing about SM-neutral fermions, it is reasonable to claim that $\mathcal{A}_{1'1'1'}$ and $\mathcal{A}_{gg1'}$ can not be used to constrain the flavon charges \cite{Weinberg1996}. For the theory to still be 
anomaly-free, we assume that the SM-neutral sector is such that these two anomalies vanish.

We again assume the SM fermion content with a type-II 2HDM and in 
addition 
three right-handed neutrinos $N_R$. 
The additional freedom from not imposing the $\mathcal{A}_{1'1'1'}$ and $\mathcal{A}_{gg1'}$ anomalies is used to impose a type-I seesaw mechanism to generate the neutrino masses 
where we assume that $N_R$ only couples to $\Phi_2$ in Eq.~(\ref{eq: L_N}). 
For the quarks and leptons we impose the same Yukawa matrices as in the anomaly-free model:
\begin{equation}\label{eq: anomalous U D L Yukawa}
Y_2^U\sim\begin{pmatrix}
\epsilon^7 & \epsilon^5 & \epsilon^3\\
\epsilon^6 & \epsilon^4 & \epsilon^2\\
\epsilon^4 & \epsilon^2 & \epsilon^0
\end{pmatrix},\ \ \ Y_1^D\sim\begin{pmatrix}
\epsilon^7 & \epsilon^6 & \epsilon^6\\
\epsilon^6 & \epsilon^5 & \epsilon^5\\
\epsilon^4 & \epsilon^3 & \epsilon^3
\end{pmatrix},\ \ \ Y_1^L\sim\begin{pmatrix}
\epsilon^8 & * & *\\
* & \epsilon^4 & \epsilon^3\\
* & \epsilon^4 & \epsilon^3
\end{pmatrix}
\end{equation}
and for the neutrinos we impose
\begin{eqnarray}\label{eq: anomalous N Yukawa}
Y_2^N\sim\begin{pmatrix}
*&*&*\\
\epsilon&1&1\\
\epsilon&1&1
\end{pmatrix},\ \ \ M_R\sim\Lambda_{FN}\begin{pmatrix}
\epsilon^2 & \epsilon & \epsilon \\
\epsilon & 1 & 1\\
\epsilon & 1 & 1
\end{pmatrix}.
\end{eqnarray}

All this is summarized in the constraints below where we have  again imposed vanishing mixing between $U(1)'$ and $U(1)_Y$ (Eq.~(\ref{eq: kinetic mixing}))
\begin{equation}\label{eq: anomalous model example}
\begin{cases}
\sum_{j=1}^3\left(Q_j^{2}-2u_j^{2}+d_j^{2}-L_j^{2}+e_j^{2}\right)=0\\
\sum_{j=1}^{3}\left(Q_j+8u_j+2d_j+3L_j+6e_j\right)=0\\
\sum_{j=1}^3\left(2Q_j+u_j+d_j\right)=0\\
\sum_{j=1}^3\left(3Q_j+L_j\right)=0\\
\sum_{j=1}^3 (2Q_j-4u_j+2d_j-2L_j+2e_j)=0\\
Q_3+u_3+H_2=0,\ 
Q_2+u_2+H_2=4,\ 
Q_1+u_1+H_2=7\\ 
Q_3+d_3-H_1=3,\ 
Q_2+d_2-H_1=5,\
Q_1+d_1-H_1=7\\ 
L_3+e_3-H_1=3,\ 
L_2+e_2-H_1=4,\ 
L_1+e_1-H_1=8\\
Q_1-Q_2=1,\
Q_2-Q_3=2\\  
L_2-L_3=0,\
L_2+\nu_2+H_2=0,\\
\nu_1=1,\ \nu_2=0,\ \nu_3=0
\end{cases}
\end{equation}
which has the Gröbner basis
\begin{equation}
\begin{array}{rrr}
Q_1 - 274/135 = 0,&\qquad Q_2 - 139/135 = 0,&\qquad Q_3 + 131/135 = 0,\\ 
u_1 - 364/135 = 0,&\qquad u_2 -  94/135 = 0,&\qquad u_3 + 176/135 = 0,\\ 
d_1 +   64/45 = 0,&\qquad d_2 +  109/45 = 0,&\qquad d_3 +  109/45 = 0,\\
L_1 + 232/135 = 0,&\qquad L_2 + 307/135 = 0,&\qquad L_3 + 307/135 = 0,\\ 
e_1 - 449/135 = 0,&\qquad e_2 +  16/135 = 0,&\qquad e_3 + 151/135 = 0,\\
\nu_1      -1 = 0,&\qquad \nu_2         = 0,&\qquad \nu_3         = 0,\\
H_1 + 863/135 = 0,&\qquad H_2 - 307/135 = 0, & 
\end{array}
\end{equation}
where we see that all  charges are directly determined. Again, note that $H_2-H_1=26/3\notin\mathbb{Z}$ so the imposed type-II $\mathbb{Z}_2$-symmetry is a residual from $U(1)'$ invariance. 
In addition, the elements that have been left unconstrained in Eqs.~(\ref{eq: anomalous U D L Yukawa}) and 
(\ref{eq: anomalous N Yukawa}) turn out to be zero.

Using the seesaw mechanism, the light neutrino mass matrix becomes
\begin{equation}
m_\nu\sim\frac{\avg{\Phi_2}^2}{\Lambda_{FN}}\begin{pmatrix}
0 & 0 & 0\\
0 & 1 & 1\\
0 & 1 & 1
\end{pmatrix}
\end{equation} 
which yields one massless neutrino and large $\nu_\mu-\nu_\tau$ mixing assuming normal hierarchy.

\section{Supersymmetric Model Example}\label{sec: supersymmetry}
Using the Froggatt-Nielsen mechanism together with supersymmetry is not a straight-forward extension of what has been done above, in particular, we start by recalling that supersymmetry and anomaly cancellation is contradictory within the Froggatt-Nielsen framework \cite{IBANEZ1994,BINETRUY1995}. 

To see this, we start with the superpotential from minimal supersymmetric SM \cite{Aitchison2007} with right-handed neutrinos: 
\begin{equation}
W=Y_{ij}^UU_j^cQ_i\cdot H_u-Y_{ij}^DD_j^cQ_i\cdot H_d-Y_{ij}^LE_j^cL_i\cdot H_d+Y_{ij}^NN_j^cL_i\cdot H_u+\frac{1}{2}M_{ij}N_i^{c}N_j^{c}+\mu H_u\cdot H_d
\end{equation}
where all fields are now superfields and there is no ``+H.c." as in the SM since supersymmetry invariance demands $W$ to be holomorphic in each of the fields. That is, for a superfield $\Psi$, $W$ is either a function of $\Psi$ or $\Psi^\dagger$, not both. This is important in the context of the Froggatt-Nielsen mechanism. In the previous cases we always had the choice of inserting $S$ or $S^*$ to balance the flavon charges, whilst now we may only use one of them. We choose to work with $S$ with flavon charge -1 following ref.~\cite{Dreiner2005}. The flavon charges of the left-handed superfields $\{Q_i,U_i^c,D_i^c,L_i,E_i^c,N_i^c,H_u,H_d\}$ are denoted as $\{Q_i,u_i,d_i,L_i,e_i,\nu_i,H_u,H_d\}$. Using the Froggatt-Nielsen mechanism the Yukawa matrices become
\begin{align}\label{eq: yukawa supersymmetry}
Y_{ij}^U&=g_{ij}^U\left(\frac{\avg{S}}{\Lambda_{FN}}\right)^{Q_i+u_j+H_u},\ \ \ 
Y_{ij}^D=g_{ij}^D\left(\frac{\avg{S}}{\Lambda_{FN}}\right)^{Q_i+d_j+H_d},\ \ \ 
Y_{ij}^L=g_{ij}^L\left(\frac{\avg{S}}{\Lambda_{FN}}\right)^{L_i+e_j+H_d}\nonumber\\
Y_{ij}^N&=g_{ij}^N\left(\frac{\avg{S}}{\Lambda_{FN}}\right)^{L_i+\nu_j+H_u},\ \ \ 
M_{ij}=g_{ij}^R\left(\frac{\avg{S}}{\Lambda_{FN}}\right)^{\nu_i+\nu_j}.
\end{align}
Note that, this means that in the supersymmetric case, the definitions of the suppression factors $n_u$ etc. are slightly different compared to earlier and now instead given by Eq.~(\ref{eq: yukawa supersymmetry}).

Imposing supersymmetry also affects the triangle anomalies since there will now be Higgsino and flavino fields contributing. The anomaly coefficients are now
\begin{equation}
\begin{cases}
\mathcal{A'}_{331'}=\mathcal{A}_{331'}\\
\mathcal{A'}_{221'}=\frac{1}{2}(H_u+H_d)+\mathcal{A}_{221'}\\
\mathcal{A'}_{111'}=2(H_u+H_d)+\mathcal{A}_{111'}\\
\mathcal{A'}_{11'1'}=2(H_u^2-H_d^2)+\mathcal{A}_{11'1'}\\
\mathcal{A'}_{1'1'1'}=2(H_u^3+H_d^3)+S^3+\mathcal{A}_{1'1'1'}+\mathcal{A}_{1'1'1'}^{\mathrm{SM-neutral}}\\
\mathcal{A'}_{gg1'}=2(H_u+H_d)+S+\mathcal{A}_{gg1'}+\mathcal{A}_{gg1'}^{\mathrm{SM-neutral}}
\end{cases}
\end{equation}
where $S=-1$ is the charge of the flavon superfield (assumed to be left-handed).

We impose the following $\epsilon$-structure for the couplings in the superpotential
\begin{align}\label{eq: matrices green-schwarz}
&Y^U\sim\begin{pmatrix}
\epsilon^7 & \epsilon^4 & \epsilon^3\\
\epsilon^6 & \epsilon^3 & \epsilon^2\\
\epsilon^4 & \epsilon^1 & \epsilon^0
\end{pmatrix},\ \ \ Y^D\sim\begin{pmatrix}
\epsilon^7 & \epsilon^6 & \epsilon^5\\
\epsilon^6 & \epsilon^5 & \epsilon^4\\
\epsilon^4 & \epsilon^3 & \epsilon^2
\end{pmatrix}\ \ \ Y^L\sim\begin{pmatrix}
\epsilon^8 & \epsilon^5 & \epsilon^4\\
\epsilon^7 & \epsilon^4 & \epsilon^3\\
\epsilon^7 & \epsilon^4 & \epsilon^3
\end{pmatrix}\nonumber\\
&M\sim\Lambda_{FN}\begin{pmatrix}
\epsilon^2 & \epsilon & \epsilon \\
\epsilon & 1 & 1\\
\epsilon & 1 & 1
\end{pmatrix}\ \ \ Y^N\sim\begin{pmatrix}
\epsilon^2&\epsilon&\epsilon\\
\epsilon&1&1\\
\epsilon&1&1
\end{pmatrix},
\end{align}
where it should be noted that, for reasons that will become clear below, the suppression factors for the $c$ and $b$ quark Yukawa couplings has been changed compared to earlier. In addition, we find that it is now also possible to completely determine the Yukawa matrices for the leptons such that they also give a PMNS-matrix with three-generation mixing.

Using the suppression factors $n_u$ etc.,  as defined by Eq.~(\ref{eq: yukawa supersymmetry}), together with the anomaly conditions gives the following two supersymmetric versions of the sum rules:
\begin{equation}\label{eq: mass product}
n_u+n_c+n_t+n_d+n_s+n_b=2\mathcal{A'}_{331'}+3(H_u+H_d)
\end{equation}
and
\begin{equation}\label{eq: mass quotient}
n_d+n_s+n_b-n_e-n_\mu-n_\tau={H_u+H_d-\left(\frac{1}{4}\mathcal{A'}_{111'}+\mathcal{A'}_{221'}-\frac{8}{3}\mathcal{A'}_{331'}\right)}.
\end{equation}
If the anomalies vanish, Eq.~(\ref{eq: mass product}) together with Eq.~(\ref{eq: matrices green-schwarz}) imply that $H_u+H_d=8$. On the other hand, vanishing anomalies together with Eqs.~(\ref{eq: mass quotient}) and (\ref{eq: matrices green-schwarz}) imply that $H_u+H_d=-1$ which directly contradicts $H_u+H_d=8$. This means that the Froggatt-Nielsen mechanism, vanishing anomalies and supersymmetry may not be joined together. 

To circumvent this, we may assume that the anomaly coefficients are non-zero but that there exists a string theoretic UV completion of the theory where the anomalies cancel via the Green-Schwarz mechanism \cite{Green1984}. For this to work we need to balance the anomaly coefficients and the so-called Kac-Moody levels $k_G$ (where $G$ labels the gauge group) below the compactification scale according to
\begin{equation}\label{eq: Green-Schwarz}
\frac{\mathcal{A'}_{111'}}{k_1}=\frac{\mathcal{A'}_{221'}}{k_2}=\frac{\mathcal{A'}_{331'}}{k_3}=\frac{\mathcal{A'}_{1'1'1'}}{3k_{1'}}=\frac{\mathcal{A'}_{gg1'}}{24}.
\end{equation}
Since the $\mathcal{A'}_{11'1'}$ anomaly can not be canceled by the Green-Schwarz mechanism we have to impose $\mathcal{A'}_{11'1'}=0$.

To obtain useful constraints out of Eq.~(\ref{eq: Green-Schwarz}) we make the standard assumption of coupling unification at the compactification scale (\cite{Maekawa2001,Dreiner2005,Chankowski2005,Kane2005}), which with our normalization of hypercharge means that $k_2=k_3$ and $k_1/k_2=20/3$. Moreover, we assume that there are SM-neutral contributions to $\mathcal{A'}_{1'1'1'}$ and $\mathcal{A'}_{gg1'}$ so that Eq.~(\ref{eq: Green-Schwarz}) is satisfied. The constraints on the flavon charges using the Green-Schwarz mechanism are thus
\begin{equation}
\begin{cases}
\mathcal{A'}_{221'}=\mathcal{A'}_{331'}\\
\mathcal{A'}_{221'}=\frac{3}{20}\mathcal{A'}_{111'}\\
\mathcal{A'}_{11'1'}=0.
\end{cases}
\end{equation}

This directly implies that $\frac{1}{4}\mathcal{A}_{111'}'+\mathcal{A}_{221'}-\frac{8}{3}\mathcal{A}_{331'}'=0$, 
which together with the suppression factors $n_i$ from Eq.~(\ref{eq: matrices green-schwarz}) gives $H_u+H_d=-1$. In turn, this
means that the $\mu$-term in the superpotential has to vanish. We note in the passing that this means that the so called $\mu$-problem may then be solved by the Giudice-Masiero mechanism \cite{Giudice1988}.  

All the constraints are summarized in the following system of equations:
\begin{equation}\label{eq: example green-schwarz}
\begin{cases}
\sum_{j=1}^3\left(Q_j^{2}-2u_j^{2}+d_j^{2}-L_j^{2}+e_j^{2}\right)+H_u^2-H_d^2=0\\
\frac{1}{2}\left[\sum_{j=1}^3\left(3Q_j+L_j\right)+H_u+H_d\right]-\frac{3}{20}\left[\frac{2}{3}\sum_{j=1}^{3}\left(Q_j+8u_j+2d_j+3L_j+6e_j\right)+2(H_u+H_d)\right]=0\\
\sum_{j=1}^3\left(2Q_j+u_j+d_j\right)-\sum_{j=1}^3\left(3Q_j+L_j\right)-H_u-H_d=0\\
Q_3+u_3+H_u=0,\ 
Q_2+u_2+H_u=3,\ 
Q_1+u_1+H_u=7\\ 
Q_3+d_3+H_d=2,\ 
Q_2+d_2+H_d=5\
Q_1+d_1+H_d=7\\ 
L_3+e_3+H_d=3,\ 
L_2+e_2+H_d=4,\ 
L_1+e_1+H_d=8\\
Q_1-Q_2=1,\
Q_2-Q_3=2\\ 
L_3+\nu_3+H_u=0,\
L_1-L_2=1,\ 
L_2-L_3=0\\
\nu_1=1,\ \nu_2=0,\ \nu_3=0
\end{cases}
\end{equation}
where we no longer have the freedom to remove the $U(1)_Y-U(1)'$ mixing in the massless limit. This system has a Gröbner basis defining a variety of just one point, which is given in Table ~\ref{table: charges green-schwarz}. 

Using a type-I seesaw mechanism, the light physical neutrino masses and mixings become:
\begin{equation}
m_\nu\sim\frac{1}{\Lambda_{FN}^2}\begin{pmatrix}
\epsilon^2 & \epsilon & \epsilon\\
\epsilon & 1 & 1\\
\epsilon & 1 & 1
\end{pmatrix},\ \ \ U_{PMNS}\sim\begin{pmatrix}
1 & \epsilon & \epsilon\\
\epsilon & 1 & 1\\
\epsilon & 1 & 1
\end{pmatrix}.
\end{equation}
This model reproduces all the fermion masses and mixings, including neutrino oscillations in three generations. 
\begin{table}
	\centering
	\caption{The unique set of rational charges satisfying Eq.~(\ref{eq: example green-schwarz}).}
	\label{table: charges green-schwarz}
	\begin{tabular}{lcccccc}
		\hline
		Generation $i$ &\quad $Q_i$ &\quad $u_i$ &\quad $d_i$ &\quad $L_i$ &\quad $e_i$ &\quad $\nu_i$\\
		\hline
		1 &\quad $\frac{103}{27}$ &\quad $\frac{128}{27}$ &\quad $\frac{71}{27}$ &\quad $\frac{23}{9}$ &\quad $\frac{44}{9}$ &\quad $1$\\
		2 &\quad $\frac{76}{27}$ &\quad $\frac{47}{27}$ &\quad $\frac{44}{27}$ &\quad $\frac{14}{9}$ &\quad $\frac{17}{9}$ &\quad 0\\
		3 &\quad $\frac{22}{27}$ &\quad $\frac{20}{27}$ &\quad $\frac{17}{27}$ &\quad $\frac{14}{9}$ &\quad $\frac{8}{9}$ &\quad 0 \\  
		\hline
		Higgs charges: &\quad $H_u=-\frac{14}{9}$ &\quad $H_d=\frac{5}{9}$ & & & &\\
		\hline
	\end{tabular}
\end{table}

\section{Summary and Conclusions}\label{sec: conclusion}
Understanding the flavor structure in the Standard Model is one of the big open questions in modern particle physics. An attractive way to explain this structure is the Froggatt-Nielsen mechanism. The new $U(1)$ charges, flavon charge, in this mechanism must satisfy both anomaly and phenomenological constraints. To find rational charges satisfying these, we have in this paper introduced methods from algebraic geometry. Especially useful is the Gröbner basis which sees and eliminates all relations among the constraints so that the system is put on its most simple and reduced form. Moreover, we discussed in detail how to deal with the case when the Gröbner basis still contains a cubic constraint and show how to find rational charges using Mordell-Weil generators. 

We have also found that the Froggatt-Nielsen constraints for the suppression of the masses are related to linear combinations of the anomaly constraints which we summarize in a set of sum rules. From these rules we conclude that the type-II (MSSM like) 2HDM is the natural setup to avoid skewed Yukawa matrices. This especially means that the type-I model, and in extension, the Standard Model, is disfavored in this setting. At the same time, 2HDM where the doublets have different charge under a $U(1)$ symmetry may possess an axion à la Weinberg and Wilczek \cite{Weinberg1978,Wilczek1978}, however, we postpone this to a later paper~\cite{Tellander2018b}. 

To conclude, using methods from algebraic geometry to study anomaly-free (vanishing anomalies or Green-Schwarz cancellation) has proven to be very useful and should be of general interest in model building. 
\section*{Acknowledgements}
We thank Martin Helmer for conversations on algebraic geometry and Joel Oredsson for many discussions. 
This work is supported in part by the Swedish Research Council, contract number 2016-05996, the European Research Council (ERC) under the European Union's Horizon 2020 research and innovation programme (grant agreement No 668679) and by the Anders Wall Foundation. 

\bibliography{Ref}

\end{document}